\documentclass[11pt]{article}
\usepackage{preamble}

\title{Nearest Neighbor Complexity and Boolean Circuits}

\author{Mason DiCicco\footnote{Worcester Polytechnic Institute. [mtdicicco@wpi.edu, dreichman@wpi.edu]},  \ Vladimir Podolskii\footnote{Tufts University. [vladimir.podolskii@tufts.edu]}, \ Daniel Reichman$^*$}

\begin{document}

\maketitle

\begin{abstract}
    
A nearest neighbor representation of a Boolean function $f$ is a set of vectors (anchors) labeled by $0$ or $1$ such that $f(\vec{x}) = 1$ if and only if the closest anchor to $\vec{x}$ is labeled by $1$. This model was introduced by \citet*{hajnal2022nearest}, who studied bounds on the minimum number of anchors required to represent Boolean functions under different choices of anchors (real vs. Boolean vectors) as well as the more expressive model of $k$-nearest neighbors.

We initiate the systematic study of the representational power of nearest and $k$-nearest neighbors through Boolean circuit complexity. To this end, we establish a close connection between Boolean functions with polynomial nearest neighbor complexity and those that can be efficiently represented by classes based on linear inequalities -- \emph{min-plus polynomial threshold functions} -- previously studied in relation to threshold circuits. This extends an observation of \cite{hajnal2022nearest}. As an outcome of this connection we obtain exponential lower bounds on the $k$-nearest neighbors complexity of explicit $n$-variate functions, assuming $k \leq n^{1-\epsilon}$. Previously, no superlinear lower bound was known for any $k>1$. 

Next, we further extend the connection between nearest neighbor representations and circuits to the $k$-nearest neighbors case. As a result, we show that proving superpolynomial lower bounds for the $k$-nearest neighbors complexity of an explicit function for arbitrary $k$ would require a breakthrough in circuit complexity. In addition, we prove an exponential separation between the nearest neighbor and $k$-nearest neighbors complexity (for unrestricted $k$) of an explicit function. These results address questions raised by \cite{hajnal2022nearest} of proving strong lower bounds for $k$-nearest neighbors and understanding the role of the parameter $k$. Finally, we devise new bounds on the nearest neighbor complexity for several explicit functions. 

\end{abstract}

\section{Introduction}
\label{sec:intro}

The capacity of a learning model to economically represent a target function is essential in designing machine learning algorithms. For example, the expressive power of machine learning models such as Boolean circuits and neural networks has received considerable attention, with a recent focus on properties that allow for efficient representations~\citep{cybenko1989approximation,hornik1989multilayer,daniely2016toward,hellerstein2007pac,eldan2016power,hajnal1993threshold,telgarsky2016benefits,martens2013representational}. Here, we study the nearest neighbor and $k$-nearest neighbors rules over the Boolean hypercube and seek to understand which Boolean functions can be represented efficiently by these rules.  

A \emph{nearest-neighbor representation} of a function $f$ is a set of vectors, called ``anchors," say $S=P \cup N$ such that $f(\vec{x})=1$ if and only if the nearest anchor to $\vec{x}$ (under the Euclidean distance metric) belongs to $P$. The set of anchors can be seen as a (disjoint) union of ``positive" and ``negative" examples. If $S \subseteq \{0,1\}^n$, we refer to the representation as \emph{Boolean}, and if $S \subseteq \RR^n$ we call it \emph{real}. This model was pioneered in \cite{hajnal2022nearest,kilic2023information}, who advocated the study of Boolean functions admitting efficient representations. In a similar vein, \cite{hajnal2022nearest} also consider the \emph{$k$-nearest neighbors} model, where the value of $f$ on input $x$ is computed as a majority vote of the $k$ nearest anchors to $x$.

The nearest neighbor rule is the subject of extensive research both with respect to statistical properties and efficient algorithms \citep{andoni2017nearest,andoni2009nearest,shakhnarovich2005nearest,dasgupta2020nearest,chaudhuri2014rates,cover1967nearest,clarkson1997nearest}. To our knowledge, much less is known about the \emph{representational efficiency} of this rule: Namely, which set of functions\footnote{We only deal with Boolean functions, but some results extend to the non-Boolean case as well, see Remark~\ref{rem:non-boolean}.} can be represented efficiently; with polynomially many anchors. 

Our approach to understanding nearest neighbor complexity is to relate it to the well-studied field of Boolean circuit complexity, following a long tradition in computational learning theory and computational complexity e.g.,~\citep{linial1993constant,jackson2002learnability,klivans2004learning,klivans2006toward,williams2018limits,vardi2021size}. Our results establish both (unconditional) lower bounds and upper bounds on the nearest neighbor complexity of explicit\footnote{By ``explicit," we mean a function that can be computed in polynomial time by a Turing machine.} functions. We also describe efficient $k$-nearest neighbors representations -- for appropriate values of $k$ -- for a large class of Boolean functions for which we do not know how to prove super-polynomial circuit lower bounds. Our proofs uncover new connections between nearest neighbors and computational models such as linear decision lists and depth-two circuits that may be of independent interest.  


\paragraph{Our results}
\label{sec:results}

To study the computational power of nearest neighbor representations, we introduce the \emph{closure} of the classes $\NN$ and $\HNN$ under substitution and duplication of variables: These are \emph{subfunctions} of Boolean functions with polynomial nearest neighbor complexity ($\NN$) as well as polynomial \emph{Boolean} nearest neighbor complexity ($\HNN$). Note that standard complexity classes based on circuit-like models of computation are closed under these operations. Hence, we can give a precise characterization (equality) of the closures of $\NN$ and $\HNN$ in terms of min-plus polynomial threshold functions of polynomial complexity ($\mpPTF$). This adds to previous results of \cite{hajnal2022nearest} showing the \emph{containment} of $\NN$ and $\HNN$ in the class of $\mpPTF$s. As a consequence, we prove (among other results) that $\HNN$ contains functions that cannot be computed by depth-two threshold circuits with polynomial size and weight. We also observe that the closure of $\HNN$ is closely connected to the class of functions with $\NN$ representations of logarithmic bit-complexity.

We study the $k$-nearest neighbors complexity of Boolean functions for $k>1$. First, we extend the aforementioned characterization -- the closure of $\NN$ in terms of $\mpPTF$ -- to the closure of $\kNN$. We use this characterization to prove that $\kNN$ for \emph{constant $k$} is closely related to $\NN$ and that there exists an explicit function that requires exponential $\kNN$ complexity when $k \leq n^{1-\epsilon}$ (for an $n$-variate function). 

Next, we generalize the characterization of $\kNN$ to \emph{arbitrary $k$} by introducing a new class, $\kSTAT$ -- functions realizable by an inequality of the $k$-statistics of two sets of linear forms -- which generalizes $\mpPTF$. Consequently, we establish that proving lower bounds for $\kNN$ for \emph{arbitrary $k$} would require a breakthrough in circuit complexity: See Theorems \ref{th:SYMMAJinkSTAT} and \ref{thm:eldl-in-kstat}. 

Finally, we present new bounds for nearest neighbor complexity of specific Boolean functions such as disjointness, CNFs, and majority. For example, we show that CNFs with polynomially many clauses have $\NN$ representations with polynomially-many anchors which also exhibit \emph{constant bit-complexity}. In contrast, there exist CNFs of polynomial size with exponential $\HNN$ complexity. We also establish a new lower bound of $n/2 + 2$ on the $\HNN$ complexity of the majority function with an even number of inputs ($n$). This matches the upper bound proved in \cite{hajnal2022nearest}.

\paragraph{Related work}
\label{sec:relatedwork}

Nearest neighbor complexity (under Euclidean distance) was formalized by \cite{hajnal2022nearest}. They observe that functions with polynomial nearest neighbor complexity can be simulated by min-plus polynomial threshold functions, but it is an open question of whether or not the inclusion $\NN \subseteq \mpPTF$ is proper. Relations to the class $\mpPTF$ are of interest because $\mpPTF$ has connections to circuit complexity. For instance, \cite{hansen2015polynomial} establish that systems of $\mpPTF$'s compute exactly the class of $\AND \circ \OR \circ \THR$ circuits.

\cite{hajnal2022nearest} prove that the functions $\THR^{\lfloor n/3 \rfloor}$\footnote{$\THR^t(\vec{x}) = \ind[x_1+\cdots+x_n \geq t]$} and $\XOR$ both require an exponential number of Boolean anchors but only $2$ and $n+1$ real anchors, respectively. In fact, the same argument proves that $\THR^{t}$ requires at least $\binom{n}{t}/\binom{2t}{t}$ anchors for any $t$, which gives exponential lower bound for $t/n$ bounded away from $1/2$, but is vacuous for $\THR^{n/2}$. It was subsequently shown in \cite{kilic2023information} that any symmetric Boolean function $f$ has an $\NN$ representation with $I(f)$ anchors, where $I(f)$ denotes the number of \emph{intervals} of $f$, and this bound is optimal when all intervals have the same length. This extends the result of \cite{hajnal2022nearest} that every symmetric function has nearest neighbor complexity at most $n+1$. 

The expressive power of $k$-nearest neighbors rule was also studied by \cite{hajnal2022nearest}. In particular, they prove that $\kNN$ can simulate \emph{linear decision trees}, which yields a linear (in $n$) lower bound for the number of anchors in $\kNN$ representations of $\IP$. They also state the open problem of proving stronger lower bounds for the $k$-nearest neighbors complexity of an explicit function.

Nearest-neighbor complexity with respect to non-Euclidean distance measures has been studied before by \cite{globig1996case, satoh1998analysis}. The former studies the problem of optimizing the distance metric itself to enable polynomial-size $\NN$ representations for a given class of Boolean functions. The latter defines ``nearest" with respect to a partial ordering based on the set of differing coordinates between anchors and, in this case, proves that CNFs and DNFs have polynomial-size representations. \cite{salzberg1995best} studies the minimum number of anchors required to represent geometric concepts in $\RR^n$. As real functions are considered, the techniques and results in~\cite{salzberg1995best} are markedly different than our work which studies Boolean functions.

It was shown in \cite{kilic2023information} that (the aforementioned) $\NN$ representations for symmetric functions have \emph{logarithmic bit-complexity}, and that this is tight for some functions. It is left as an open problem to characterize $\NN$ representations of \emph{threshold functions} in terms of bit-complexity. Other works have studied the role of bit-complexity in \emph{approximate nearest neighbor search}; where we wish to find an anchor whose distance is minimal to a query point, up to a factor of $(1+\epsilon)$. For example, ~\cite{indyk2018approximate} provide tight bounds (in terms of bit-complexity) on the size of data structures performing approximate nearest neighbor search. This setting is quite different from our focus on exact classification of Boolean vectors. 

The bit-complexity of the weights in polynomial-size threshold circuits has been studied extensively (see, e.g., surveys~\citep{razborov1992small,Saks_1993}). For example, it was proved by \citep{goldmann1992majority,goldmann1998simulating} that arbitrarily large weights can be reduced to have logarithmic bit-complexity by slightly increasing the depth of the circuit (along with a polynomial blow-up in size). 

Constructions of Boolean circuits computing nearest neighbor classification are known. \cite{murphy1992nearest} constructs the ``canonical'' $\OR \circ \AND \circ \THR$ circuit computing any function with an $m$-anchor $\NN$ representation in size $O(m^2)$. (See Appendix \ref{sec:circuitUB}). A very similar depth-three construction for $\kNN$, also with size $O(m^2)$, was found by \cite{chen1995implementing}. Note that the weights of the above circuits are bounded by a polynomial in $n$. 

There is extensive research on the learning complexity -- in particular, statistical parameters such as sample complexity -- of \emph{real} nearest neighbor representations (e.g., \cite{devroye2013probabilistic}). There is evidence that increasing $k$ in the $k$-nearest neighbors rule can decrease its estimation error \cite{devroye1981asymptotic}. In addition, the algorithmic complexity of nearest neighbors was considered by \cite{wilfong1991nearest}, who establishes NP-hardness of finding a set of anchors with minimal cardinality within some ground set (i.e., the \emph{consistent subset} problem). 

\paragraph{Organization}

Section \ref{sec:prelim} outlines basic definitions required in subsequent sections. Section \ref{sec:mpPTFvsNN} establishes the equivalence between $\closHNN,\closNN$ and min-plus polynomial threshold functions, then discusses some of the consequences. Section \ref{sec:kNNvskSTAT} generalizes $\mpPTF$ to a new class, $\kSTAT$, and proves that a similar equivalence holds with $\closkNN$. Here, we also derive several connections to circuit classes such as $\SYM \circ \MAJ$. Section \ref{sec:NNcomplexity} contains several new results (upper and lower bounds) for the nearest neighbor complexity of explicit Boolean functions. Many proofs are relegated to Appendix \ref{sec:proofs} due to space constraints. Appendix \ref{sec:circuitUB} contains some direct constructions of threshold circuits computing $\HNN$, one of which has depth two.

\section{Preliminaries}
\label{sec:prelim}
The use the following notation throughout the paper:

\medskip
\begin{tabular}{|l}
Vectors are written in bold (i.e. $\vec{x} = (x_1,\cdots,x_n)$). \\
\parbox{\textwidth}{The $k$'th statistic of $\vec{x}$, denoted $\vec{x}_{(k)}$, is the $k$'th smallest element of $\vec{x}$. 

--- In particular, $\vec{x}_{(k)} = x_{\sigma(k)} \iff x_{\sigma(1)} \leq \cdots \leq x_{\sigma(n)} \text{ for some } \sigma \in S_n$}  \\
$\Delta(\vec{x},\vec{y}) := ||\vec{x}-\vec{y}||^2_2$ denotes the distance between $\vec{x},\vec{y}$. \\
$\langle \vec{x},\vec{y} \rangle$ denotes the real inner product (dot product), $x_1y_1 + \cdots + x_ny_n$. \\
$\poly(n)$ refers to an arbitrary polynomial in the variable $n$. \\
$\ind[P]$ denotes the Boolean function whose value is 1 if and only if $P$ holds. \\
\end{tabular}

\medskip
Note that the (squared) Euclidean distance between two Boolean vectors is equal to their Hamming distance, $\Delta(\vec{x},\vec{y}) = \sum_{i \leq n}  \ind[x_i \neq y_i]$, so the Hamming weight of a Boolean vector $\vec{p}$ is denoted $\Delta(\vec{p}) := \Delta(\vec{p} , \vec{0}) = ||\vec{p}||_2^2$. 

\subsection{Boolean functions}
\label{sec:boolfunc}
\begin{definition}
A \emph{threshold gate} is a Boolean function $f : \{0,1\}^n \to \{0,1\}$ defined by a weight vector $\vec{w} \in \RR^n$ and a threshold $\theta \in \RR$ such that
\begin{equation}
\label{eq:THR}
f(\vec{x}) = 1 \iff \langle \vec{w}, \vec{x}\rangle \geq \theta.
\end{equation}

A \emph{threshold circuit} is a sequence $(f_1,\cdots,f_s)$ of $s \geq n$ gates such that the first $n$ gates are equal to the input variables (i.e., $f_i = x_i$ for $i \leq n$) and subsequent gates are threshold gates whose inputs are some subset of the previous gates. The output of the circuit is equal to the output of the final gate(s). The size of the circuit is equal to $s-n$. 
\end{definition}

A threshold circuit can be viewed as a directed acyclic graph. Nodes with fan-in 0 correspond to inputs, and other nodes correspond to threshold gates applied to the values of the preceding nodes. Nodes with fan-out 0 correspond to output nodes. The \emph{depth} of the circuit is the length of the longest path from an input node to an output node.

\begin{remark}
It is well known that we may assume that the weights (and the threshold) are integers without loss of generality: Since the domain of a threshold gate is finite, we may approximate each weight by a rational number and multiply by a common denominator. See \cite{jukna2012boolean} for a comprehensive introduction to circuit complexity.
\end{remark}

\begin{definition}
\label{def:HNN}
A \emph{Nearest Neighbor ($\NN$)} representation of a Boolean function $f : \{0,1\}^n \to \{0,1\}$ is defined by two disjoint sets of positive and negative \emph{anchors} $P,N \subseteq \RR^n$ such that
\begin{itemize}
\item $f(\vec{x})=1$ if there exists a $\vec{p} \in P$ with $\Delta(\vec{x},\vec{p}) < \Delta(\vec{x},\vec{q})$ for all $\vec{q} \in N$.
\item $f(\vec{x})=0$ if there exists a $\vec{q} \in N$ with $\Delta(\vec{x},\vec{q}) < \Delta(\vec{x},\vec{p})$ for all $\vec{p} \in P$.
\end{itemize} 

A \emph{Hamming Nearest Neighbor ($\HNN$)} representation is defined identically for Boolean anchors.

\end{definition}

\begin{definition}
\label{def:kNN}
A \emph{$k$-Nearest Neighbors ($\kNN$)} representation of a function $f : \{0,1\}^n \to \{0,1\}$ is defined by two disjoint sets of positive and negative \emph{anchors} $P,N \subseteq \RR^n$ and an integer $k$ such that
\begin{center}
$f(\vec{x})=1$ $\iff$ there exists an $A \subseteq P \cup N$ with the following properties:
\begin{enumerate}
    \item $|A| = k$
    \item $|A \cap P| \geq |A \cap N|$
    \item $\Delta(\vec{x},\vec{a}) < \Delta(\vec{x},\vec{b})$ for all $\vec{a} \in A$, $\vec{b} \not \in A$.
\end{enumerate}
\end{center} 
A $\kHNN$ representation is defined identically for Boolean anchors. 
\end{definition}

\begin{definition}[\cite{hansen2015polynomial}]
A \emph{min-plus polynomial threshold function ($\mpPTF$)} is a Boolean function $f : \{0,1\}^n \to \{0,1\}$ defined by two sets of linear forms with integer coefficients\footnote{As for threshold gates, there is no  loss of generality in the assumption that weights of $\mpPTF$s are integers.
} $\{L_{1},\cdots, L_{\ell_1}\} \cup \{R_{1},\cdots, R_{\ell_2}\}$ satisfying
\begin{equation}
\label{eq:mpPTF}
f(\vec{x})=1 \iff \min_{i \leq \ell_1} L_{i}(\vec{x}) \leq \min_{j \leq \ell_2} R_{j}(\vec{x})
\end{equation}

The number of \emph{terms} in an $\mpPTF$ is equal to $\ell_1 + \ell_2$, and the \emph{maximum weight} is equal to the largest absolute value of the coefficients of any form.
\end{definition}

\begin{definition}[\cite{rivest1987learning}]
A \emph{linear decision list ($\LDL$)} representation of a Boolean function $f$ is a sequence of instructions ``if $f_i(\vec{x}) = 1$, then output $c_i$'' for $1 \leq i \leq m$, followed by ``output 0.'' Here, $f_1,\cdots,f_m$ are threshold gates and $c_1,\cdots,c_m \in \{0,1\}$. \emph{Exact} linear decision lists ($\ELDL$) are defined similarly using \emph{exact} threshold functions -- threshold gates where the inequality in \eqref{eq:THR} is replaced with equality. The \emph{length} of an $\LDL$ or $\ELDL$ is the number of gates, $m$, and its \emph{maximum weight} is equal to the largest coefficient of any $f_i$.

\end{definition}

\begin{definition}
\label{def:funcs}
We consider the following well-known Boolean functions.

\begin{tabular}{|l}
    The \emph{majority} function, $\MAJ(x_1,\cdots,x_n) = \ind[x_1 + \cdots + x_n \geq n/2]$ \\
    The \emph{disjointness} function, $\DISJ(\vec{x},\vec{y}) = \ind[\langle \vec{x}, \vec{y} \rangle = 0]$ \\
    The \emph{inner product mod 2} function, $\IP(\vec{x},\vec{y}) = \langle \vec{x}, \vec{y} \rangle \mod 2$ \\
    The \emph{odd-max-bit} function, $\OMB(x_1,\cdots,x_n) = \max\{i:x_i=1\}  \mod 2$ \\
\end{tabular}

\end{definition}

\subsection{Function classes}
\label{sec:classes}

First, we define classes of Boolean circuits whose inputs may be variables, their negations, or the constants 0 and 1. $\AND$, $\OR$, $\THR$, and $\SYM$ are the classes of polynomial-size\footnote{``polynomial'' in this context is always with respect to the input size, $n$.} depth-one circuits composed of $\AND$, $\OR$, threshold gates, and symmetric functions (i.e., Boolean functions which depend only on the Hamming weight of the input) respectively. $\MAJ \subset \THR$ is the set of threshold gates with polynomial weights\footnote{We abuse the notation denoting by $\MAJ$ both specific function and a class of function. The meaning of our notation will be also clear from the context.}. $\AC^0$ is the class of constant-depth circuits consisting of a polynomial number of $\AND$, $\OR$, and $\NOT$ gates. A CNF is an $\AND$ of clauses, where each clause is an $\OR$ of (possibly negated) variables.

For two circuit classes $C_1$, $C_2$, the class of circuits consisting of a circuit from $C_1$ whose inputs are (the outputs of) a polynomial number of circuits from $C_2$ is denoted by $C_1 \circ C_2$. (e.g., $\THR \circ \THR$ refers to depth two threshold circuits of polynomial size.)

\begin{definition}
$\NN$ is the class of Boolean functions that have nearest neighbor representations with polynomially-many anchors. $\HNN$ is the same class where anchors are Boolean. $\kNN$ and $\kHNN$ are defined in the same manner for a positive integer $k$.
\end{definition}

\begin{definition}
$\mpPTF(\infty)$ is the class of min-plus polynomial threshold functions with a polynomial number (in terms of the number of inputs) of terms and unbounded maximum weight. $\mpPTF(\poly(n))$ is the same class with polynomially-bounded maximum weight.
\end{definition}

\begin{definition}
$\LDL$ is the class of Boolean functions representable by linear decision lists with polynomial length. $\LDLlw$ is the same class with polynomially-bounded maximum weight. $\ELDL$ and $\ELDLlw$ are defined similarly for exact linear decision lists.
\end{definition}

\section{\texorpdfstring{Min-plus PTFs vs. nearest neighbors}{}}
\label{sec:mpPTFvsNN}

In this section, we introduce the \emph{closure} operation and derive an equivalence between (the closure of) $\NN$, $\HNN$ and $\mpPTF$.

\begin{definition}
Define a substitution of variables as a function $v : \{0,1\}^n \rightsquigarrow \{0,1\}^{\widetilde{n}}$ where $\rightsquigarrow$ \emph{duplicates} variables or \emph{adds constant} variables (e.g., $x_1x_2 \rightsquigarrow x_1x_1x_2x_2x_20$). Then, a Boolean function $f : \{0,1\}^n \to \{0,1\}$ is a \emph{subfunction} of $g : \{0,1\}^{\widetilde{n}} \to \{0,1\}$ if $\widetilde{n} = \poly(n)$ and there exists a substitution of variables $v$  such that $f(\vec{x}) = g(v(\vec{x}))$ for all $\vec{x} \in \{0,1\}$.
\end{definition}

Subfunctions may equivalently be obtained from $g : \{0,1\}^{\widetilde{n}} \to \{0,1\}$ by \emph{identifying variables} (e.g., $x_1 = x_2$) and \emph{assigning variables} to constants (e.g., $x_1=0$).

\begin{definition}
For any function class $C$, let $\overline{C}$ denote the \emph{closure} of $C$: The set of subfunctions derived from the elements of $C$\footnote{Note that classes based on circuits are already closed under this operation.}.

\medskip
In particular, we say that a Boolean function $f$ has an ``$\closNN$ representation" if it is a subfunction of some $g \in \NN$.
\end{definition}

\begin{theorem}
\label{th:mpPTFequalsNN}
\[
\closNN = \mpPTF(\infty), \ \ \ 
\closHNN = \mpPTF(\poly(n))
\]
\end{theorem}

Theorem \ref{th:mpPTFequalsNN} and some consequences are proved in Appendix \ref{sec:proofs}. Namely, we observe that any $n$-variate function in $\closNN$ is a sub-function of an $(n+1)$-variate $\NN$ representation, and that $\mpPTF(\poly(n))$ captures precisely the power of $\closNN$ representations with \emph{bit-complexity} $O(\log n)$. Then, using the results of \cite{hajnal2022nearest} and \cite{hansen2015polynomial}, we immediately establish the following two corollaries.

\begin{corollary}
\label{cor:HNNneqHNNbar}
$\HNN \subsetneq \closHNN$
\end{corollary}

\begin{corollary}
\label{cor:signrankLB}
$\closNN$ representations of $\IP$, $f_n(\vec{x},\vec{y}) := \bigwedge_{i=1}^n \bigvee_{j=1}^{n^2} (x_{i,j} \wedge y_{i,j})$ require $2^{\Omega(n)}$ anchors.
\end{corollary}

Theorem \ref{th:mpPTFequalsNN} also yields lower bounds for the circuit complexity of functions belonging to $\HNN$. (A direct construction in Appendix \ref{sec:circuitUB} shows that $\HNN \subseteq \THR \circ \MAJ$.)

\begin{theorem}
\label{th:HNNnotMAJMAJ}
\[
\HNN \not \subseteq \MAJ \circ \MAJ 
\]
More precisely, there is a Boolean function with an $\HNN$ representation with $n+1$ anchors which cannot be computed by a depth-two majority circuit with $\poly(n)$ gates.
\end{theorem}

\begin{proof}
First, we claim that $\OMB \circ \AND_2 \in \closHNN$. Indeed, $f$ is computed by an $\mpPTF$ with $n+1$ terms:
\[
\min\{L_1(\vec{x},\vec{y}),L_3(\vec{x},\vec{y}),\cdots\} \leq \min\{-1,L_2(\vec{x},\vec{y}),L_4(\vec{x},\vec{y}),\cdots\} 
\]
where $L_k(\vec{x},\vec{y}) = (k+1) \cdot (1-x_i-y_i)$. Note that if $x_i=y_i=1$, then $L_i(\vec{x},\vec{y}) = -(i+1)$ and otherwise $L_i(\vec{x},\vec{y}) \geq 0$. Hence, the minimum is obtained at the maximum index $j$ where $x_j=y_j=1$. The claim follows from Theorem \ref{th:mpPTFequalsNN}.

Second, it is known that $\OMB \circ \AND_2 \not \in \MAJ  \circ \MAJ$ by \citep{buhrman2007computation, hajnal1993threshold}. Thus, if $\HNN$ was in $\MAJ  \circ \MAJ$, then we could use the $\closHNN$ representation described above to get a $\MAJ \circ \MAJ$ circuit computing $\OMB \circ \AND_2$, which is a contradiction.
\end{proof}

Finally, we observe a connection between $\mpPTF$s and linear decision lists. This provides additional proof techniques for $\closHNN$ and helps to relate a question of separation of $\closHNN$ and $\closNN$ to the similar question for linear decision lists. The following lemma is proved in Appendix \ref{sec:proofs}.

\begin{lemma} \label{lem:PTF-to-LDL}
\[
\mpPTF(\poly(n)) \subseteq \LDLlw.
\]
More precisely, any $\mpPTF$ with $m$ terms and maximum weight $W$ is equivalent to a linear decision list with length and maximum weight $O(m^2W)$. 
\end{lemma}

\begin{remark}
This lemma enables another technique to prove lower bounds for $\closHNN$ apart from sign-rank. More specifically, it is known that any function without large monochromatic rectangles must have a large linear decision list by \cite{ChattopadhyayMM20linear-decision-lists}.
\end{remark}

\begin{lemma} \label{lem:LDL-to-PDF}
    $\LDL \subseteq \mpPTF(\infty)$.
\end{lemma}

\begin{proof}
    It was shown in~\cite[Lemma 22]{hansen2015polynomial} that $\OMB \circ \THR \subseteq \mpPTF(\infty)$. Our lemma follows since $\OMB$ is complete for the class of decision lists -- See \cite[Lemma 22]{hansen2015polynomial}.
\end{proof}

It is open whether $\LDLlw$ and $\LDL$ are equal by \cite{ChattopadhyayMM20linear-decision-lists}. Lemmas \ref{lem:PTF-to-LDL} and \ref{lem:LDL-to-PDF} immediately allow us to relate this problem to the problem of separating $\closHNN$ and $\closNN$.

\begin{corollary}
    If $\LDLlw \neq \LDL$, then $\closHNN \neq \closNN$.
\end{corollary}

\begin{proof}
    From Theorem \ref{th:mpPTFequalsNN} and Lemmas~\ref{lem:PTF-to-LDL},~\ref{lem:LDL-to-PDF}, we have the following sequence of inclusions.
    \[
    \closHNN = \mpPTF(\poly) \subseteq \LDLlw \subseteq \LDL \subseteq \mpPTF(\infty) = \closNN,
    \]
    If $\closHNN = \closNN$, then the whole sequence of inclusions collapses and, in particular, $\LDLlw = \LDL$.
\end{proof}

\section{\texorpdfstring{$\kNN$ vs. Circuits}{}}
\label{sec:kNNvskSTAT}

In this section, we give a circuit-style characterization of $\kNN$ and provide connections to known circuit classes. From these results, we obtain a separation between $\kNN$ and $\NN$. Additionally, our results imply complexity theoretic barriers for proving superpolynomial lower bounds for $\kNN$ representations of explicit functions.

\subsection{\texorpdfstring{Characterization for small $k$}{}} 
\label{sec:constantkNN}

Here, we use the connection to $\mpPTF$ representations to get our first results on $k$-nearest neighbors complexity. In particular, we establish connection of $k$-nearest neighbors representations for constant $k>0$ to $\closNN$ and prove a lower bound on $k$-nearest neighbors complexity for sublinear $k$. 

\begin{theorem}
\label{th:kNNmpPTF}
Any Boolean function with an $m$-anchor $\kNN$ representation is computed by an $\mpPTF$ with $\binom{m}{k}$ terms.
\end{theorem}

\begin{proof}
We prove only the first statement as both arguments are identical. As noted in the proof of Theorem \ref{th:mpPTFequalsNN}, the distances from anchors to a query point $\vec{x}$ are linear forms $L_1(\vec{x}),\cdots,L_m(\vec{x})$. Assign each linear form a label $\ell_1,\cdots,\ell_m \in \{1,-1\}$ where a positive label indicates placement on the left-hand side of the $\mpPTF$ and vice versa. 

Then, consider the collection $A^+(\vec{x}) = \{ L_{i_1}(\vec{x}) + \cdots + L_{i_k}(\vec{x}) \ | \ \ell_{i_1} + \cdots + \ell_{i_k} \geq 0\}$ and the compliment $A^-(\vec{x}) = \{ L_{i_1}(\vec{x}) + \cdots + L_{i_k}(\vec{x}) \ | \ \ell_{i_1} + \cdots + \ell_{i_k} < 0\}$. The resulting $\mpPTF$ with $\binom{m}{k}$ terms, $\ind[\min A^+(\vec{x}) \leq \min A^-(\vec{x})]$, realizes the original $\kNN$ representation: The minimum is attained by groups of $k$-nearest neighbors and if any such group has a positive majority then the inequality holds.
\end{proof}

It follows that Boolean functions with $m$-anchor $\kNN$ representations can be represented in $\overline{\NN}$ with $\binom{m}{k}$ anchors. These results generalize to both weighted $\kNN$ and to non-Boolean inputs. See Appendix~\ref{sec:proofs} for a discussion.

As a consequence of Theorem \ref{th:kNNmpPTF}, sign-rank lower bounds (e.g., Corollary \ref{cor:signrankLB}) also apply to $\kNN$. In particular, we get an exponential lower bound for $\kNN$ with $k = O(n^{1- \epsilon})$ for constant $\epsilon>0$. This addresses an open question posed in~\cite{hajnal2022nearest} regarding $k$-nearest neighbors complexity.

\begin{corollary}
\label{cor:kNNsignrank}
Any $\closkNN$ representation of $\IP$ or $f_n(\vec{x},\vec{y}) := \bigwedge_{i=1}^n \bigvee_{j=1}^{n^2} (x_{i,j} \wedge y_{i,j})$ requires $2^{\Omega(n/k)}$ anchors\footnote{The constant in $\Omega$-notation does not depend on $k$ or $n$.}. 
\end{corollary}

\begin{proof}
Assume that $\IP$ (or $f_n$) has a $\closkNN$ representation with $m$ anchors. By Theorems \ref{th:mpPTFequalsNN} and \ref{th:kNNmpPTF} , $\IP$ has an $\closNN$ representation with $\binom{m}{k} \leq m^k$ anchors. By Corollary~\ref{cor:signrankLB}, we have $m^k \geq 2^{\Omega(n)}$ and thus $m \geq 2^{\Omega(n/k)}$.
\end{proof}

\subsection{\texorpdfstring{Characterization for arbitrary $k$}{}} 

In this section, we generalize the ideas of Theorem \ref{th:mpPTFequalsNN} to the closure of $\kNN$, yielding further connections between nearest neighbors and circuit complexity. All proofs and further discussion can be found in Appendix \ref{sec:proofs}.

\begin{definition} 
\label{def:kSTAT}
Define by $\kSTAT$ the class of functions $f \colon \{0,1\}^n \to \{0,1\}$ representable by an inequality between $k$-statistics of two sets consisting of a polynomial number of linear forms: Given $\{L_{1}, \cdots, L_{\ell_1}\} \cup \{R_{1}, \cdots, R_{\ell_2}\}$ and integers $k_l$ and $k_r$, 
\begin{equation}
\label{eq:kSTAT}
f(\vec{x}) = 1 \iff (L_{1}(\vec{x}), \cdots, L_{\ell_1}(\vec{x}))_{(k_l)} < (R_{1}(\vec{x}), \cdots, R_{\ell_2}(\vec{x}))_{(k_r)}
\end{equation}
and $\ell_1+\ell_2$ is bounded by a polynomial in $n$.

As usual, we can assume that all coefficients in the linear forms are integers. Define the subclass $\kSTATlw$ where all coefficients are bounded by a polynomial in $n$\footnote{$\mpPTF$ can be viewed as a special case of $\kSTAT$ in which $k_l=k_r=1$.}.
\end{definition}

Note that we can reduce Definition~\ref{def:kSTAT} to the case of $k_l=k_r$ with only a linear increase in the size. This can be done by adding ``dummy'' linear forms that are always smaller than all others.

\begin{theorem}
\label{th:kNNequalskSTAT}
\[
\closkNN = \kSTAT, \ \closkHNN = \kSTATlw.
\]
\end{theorem}

Next, we provide another equivalent form of $\kSTAT$ that is sometimes more convenient.

\begin{theorem} \label{thm:altern-k-stat}
    The class $\kSTAT$ consists exactly of functions $f \colon \{0,1\}^n \to \{0,1\}$
    for which there exist linear forms $\{L_1, \cdots, L_p\}$ with $p = \poly(n)$, a positive integer $k$, and a labelling function $\lab \colon \{1,\cdots,p\} \to \{0,1\}$, such that for all $\vec{x}$, 
    \begin{equation}
    \label{eq:altkSTAT}
    f(\vec{x})=1 \iff (L_1(\vec{x}), \cdots, L_p(\vec{x}))_{(k)} = L_i(\vec{x}) \text{ for some $i$ with } \lab(i)=1.
    \end{equation}
    The class $\kSTATlw$ consists exactly of functions with the same representation with polynomial-size coefficients in the linear forms.
\end{theorem}

Now we show that some well-known circuit classes, for which we do not have any known lower bounds, are computable by $\closkHNN$.

\begin{theorem}
\label{th:SYMMAJinkSTAT}
    $\SYM \circ \MAJ \subseteq \kSTATlw$.
\end{theorem}

Using the same strategy, we can embed a large complexity class into $\kNN$ directly:

\begin{theorem}
\label{th:SYMANDinkNN}
    $\SYM \circ \AND \subseteq \kNN$.
\end{theorem}

\begin{remark}
Note that $\SYM \circ \AND \subseteq \SYM \circ \MAJ$ and $\SYM \circ \AND$ is known to simulate the whole class of $\ACC^0$ within quasi-polynomial size~\cite{BeigelT94ACC}. Related classes are of interest in the context of obtaining lower bounds through circuit satisfiability algorithms~\cite[Conjecture 1]{VyasW23SYM}. 
\end{remark}

As a result of Theorem~\ref{th:SYMANDinkNN}, if we prove for some explicit function $f$ that $f \notin \kNN$, it will follow that $f \notin \SYM \circ \AND$, and this would be a major breakthrough in circuit complexity. Also note that $\IP \in \SYM \circ \AND$ and thus, by Theorem~\ref{th:SYMANDinkNN}, $\IP \in \kNN$. Together with Corollary~\ref{cor:signrankLB}, this gives a separation between $\NN$ and $\kNN$. This also shows that in Corollary~\ref{cor:kNNsignrank} we cannot get rid of $k$ in the lower bound.

\begin{theorem} \label{thm:eldl-in-kstat}
    $\ELDL \subseteq \kSTAT$, $\ELDLlw \subseteq \kSTATlw$.
\end{theorem}

\begin{remark}
The class $\ELDL$ is known to be contained in $\THR \circ \THR$ and proving super-polynomial lower bounds for $\ELDL$ is an open problem (See \cite{DahiyaVMS24linear-decision-lists}).  
\end{remark}

\section{New bounds for the nearest neighbor complexity of Boolean functions}
\label{sec:NNcomplexity}

In this section, we derive several bounds on the nearest neighbor complexity of Boolean functions.

\subsection{Nearest neighbor complexity of CNFs}

We first show that any CNF admits an efficient $\NN$ representation.

\begin{theorem}
\label{th:AC2NN}
Any CNF or DNF with $m$ clauses has an $\NN$ representation with $m + 1$ anchors and constant bit-complexity.
\end{theorem}

\begin{proof}
It suffices to prove the statement for DNFs as any CNF can be converted to a DNF by negation.

\medskip
Let $N = \{\vec{q} := (\frac12,\cdots,\frac12)\}$ and note that $d(\vec{x},\vec{q}) = n/4$ for every input $\vec{x} \in \{0,1\}^n$ (where $d$ is the squared Euclidean distance). For each clause, say $C(\vec{x}) = (x_1 \wedge \cdots \wedge x_k)$, introduce a positive anchor 
\[
\vec{p_C} = \bigg(\underbrace{1,\frac32,\cdots,\frac32}_{k},\underbrace{\frac12,\cdots,\frac12}_{n-k}\bigg)
\]
If any variable is negated, replace the corresponding $\frac{3}{2}$ (or $1$) with $-\frac{1}{2}$ (or $0$). 

\medskip
If $C(\vec{x})=1$, then $d(\vec{x}, \vec{p_C}) = (n-1)/4  < d(\vec{x},\vec{q})$. Otherwise, some literal in $C$ is equal to zero, hence $d(\vec{x}, \vec{p_C}) \geq 1 + (n-1)/4  > d(\vec{x}, \vec{q})$. Therefore, the entire DNF, say $C_1 \vee \cdots \vee C_m$, is satisfied if and only if some $\vec{p_{C_i}}$ is a nearest neighbor of $\vec{x}$. 

\end{proof}

The compact representation above does not generalize to deeper $\AC^0$ circuits of depth larger than $2$. For instance, Corollary~\ref{cor:signrankLB} exhibits a function computable by a depth-three De Morgan circuit of polynomial size which does not belong to $\closNN$. For the well studied disjointness function (that admits a compact CNF representation) we can get an efficient $\HNN$ representation:

\begin{theorem}
\label{th:DISJinHNN}
\[
\DISJ \in \HNN
\]
The disjointness function (in $2n$ dimensions) has an $\HNN$ representation with $3n$ anchors.
\end{theorem}

\begin{proof}
Consider anchors $P = \{(\vec{e_1},\vec{e_1}),\cdots,(\vec{e_n},\vec{e_n})\}$ and $N = \{\vec{e_1},\cdots,\vec{e_{2n}}\}$ where $\vec{e_i}$ denotes the $i$'th standard basis vector and $(\vec{e_i},\vec{e_i})$ their concatenation.

Let $\vec{x},\vec{y} \in \{0,1\}^n$ and suppose $x_i=y_i=1$ for some $i$. Then, for all $j$ it holds that $\Delta((\vec{x},\vec{y}), (\vec{e_i},\vec{e_i})) \leq \Delta((\vec{x},\vec{y}), \vec{e_j}) - 1$ with equality when $i=j$. Otherwise, $\Delta((\vec{x},\vec{y}), (\vec{e_i},\vec{e_i})) \geq \Delta((\vec{x},\vec{y}), \vec{e_j}) + 1$ for all $i,j$.
\end{proof}

\begin{remark}
It can be shown that the number of anchors in Theorem~\ref{th:DISJinHNN} is nearly tight; based on the $\Omega(n)$ lower bound for $\DISJ$ of \cite{razborov1990distributional}, a simple argument proves that $\NN$ representations of disjointness require $\Omega(n/ \log n)$ anchors. We omit the details.     
\end{remark}

Proceeding, we show that some CNFs with polynomially many clauses have exponential \emph{Boolean} nearest neighbor complexity.

\begin{definition}
The \emph{Hamming cube graph} is an undirected graph with vertices $V = \{0,1\}^n$ and edges $E = \{(\vec{u},\vec{v}) \in V : \Delta(\vec{u},\vec{v})=1\}$. The \emph{components} of a Boolean function $f$ are the connected components of the subgraph of the Hamming cube graph induced by the vertex set $f^{-1}(1)$.
\end{definition} 

\begin{lemma}
\label{lem:HNNcomponents}
If a Boolean function $f$ has $m$ components then any $\HNN$ representation of $f$ has at least $m$ anchors.
\end{lemma}

\begin{proof}
Consider some component $C$ of $f$ and let $\delta(C)$ denote the vertex boundary of $C$: Vertices in $\{0,1\}^n \setminus C$ with a neighbor in $C$. Note that $\delta(C) \subseteq f^{-1}(0)$.

Suppose $f$ has $\HNN$ representation $P \cup N$ and let $\vec{p} \in P$ be the nearest anchor to some $\vec{x} \in C$. Assume for contradiction that $\vec{p} \not \in C$. Note that $\Delta(\vec{x},\vec{p})$ is equal to the length of the shortest path from $\vec{x}$ to $\vec{p}$ in the Hamming cube graph, which by assumption must contain some $\vec{y} \in \delta(C)$. (In particular, $\Delta(\vec{x}, \vec{p}) = \Delta(\vec{x},\vec{y}) + \Delta(\vec{y},\vec{p})$.) Thus, there must exist some negative anchor $\vec{q} \in N$ with $\Delta(\vec{y},\vec{q}) < \Delta(\vec{y}, \vec{p})$. By the triangle inequality,
\[
\Delta(\vec{x},\vec{q}) \leq \Delta(\vec{x},\vec{y}) + \Delta(\vec{y},\vec{q}) < \Delta(\vec{x},\vec{y}) + \Delta(\vec{y},\vec{p}) = \Delta(\vec{x}, \vec{p})
\]
which contradicts the minimality of $\vec{p}$. Thus, each component contains an anchor.
\end{proof}

Using the previous results, another separation between $\HNN$ and $\NN$ follows from the existence of a CNF (over $n$-variables) with $\poly(n)$ clauses and exponentially (in $n$) many components. (See Appendix \ref{sec:proofs}.)

\begin{theorem}
\label{th:CNFnotinHNN}
For any $k>0$, there exists a $k$-CNF over $n$ variables with $\poly(n)$ clauses for which any $\HNN$ representation has $2^{\Omega(n)}$ anchors. 
\end{theorem}

\subsection{A new lower bound for majority}
\label{sec:ties}
We now discuss the disparity between the $\HNN$ complexity of the majority function in~\cite[Theorem 4]{hajnal2022nearest}: In particular, when $n$ is even, the best upper bound is $\frac{n}{2}+2$ anchors, whereas $2$ anchors suffices when $n$ is odd. Note that if ties were allowed (won by positive anchors) in Definition \ref{def:HNN}, then $P = \{1^n\}$ and $N = \{0^n\}$ would suffice as an $\HNN$ representation for $\MAJ$ for all $n$.

\begin{theorem} \label{thm:maj-n-2-lower-bound}
For even $n$, any $\HNN$ representation of $\MAJ$ requires $\frac{n}{2}+2$ anchors.
\end{theorem}

\begin{proof}
Suppose $P \cup N$ is an $\HNN$ representation of $\MAJ$ for even $n$. We claim that for each $\vec{x} \in \{0,1\}^n$ satisfying $\Delta(\vec{x})=n/2$, there is a positive anchor $\vec{p} \neq \vec{1}$ with $\vec{x} \leq \vec{p}$ in coordinate-wise order: 

It follows from \cite{hajnal2022nearest} that the nearest anchor $\vec{p}$ to $\vec{x}$ satisfies $\vec{x} \leq \vec{p}$. Indeed, for some $i$ it holds that $x_i = 1$, so suppose for contradiction that $p_i = 0$. Then, construct $\vec{y} = \vec{x} - \vec{e_i}$ and let $\vec{q} \in N$ be the nearest anchor to $\vec{y}$. This yields $\Delta(\vec{x}, \vec{p}) = \Delta(\vec{y}, \vec{p}) + 1 > \Delta(\vec{y}, \vec{q}) + 1$, contradicting the fact that 
\begin{equation}
\label{eq:pcloserthanq}
\Delta(\vec{x}, \vec{p}) < \Delta(\vec{x}, \vec{q}) \leq \Delta(\vec{y}, \vec{q}) + 1.
\end{equation}
 A similar argument shows that $\vec{q} \leq \vec{y}$. Hence, $\Delta(\vec{y},\vec{q}) \leq \frac{n}{2}-1$, and \eqref{eq:pcloserthanq} becomes $\Delta(\vec{x},\vec{p}) < \frac{n}{2}$ which implies that $\Delta(\vec{p}) \leq n-1$, proving the claim.

\medskip
For contradiction, assume that $|P \cup N| \leq \frac{n}{2} + 1$. Since there must be at least one negative anchor, we have $|P| \leq \frac{n}{2}$. Then, we can construct $\vec{x} \in \{0,1\}^n$ with $\Delta(\vec{x}) = \frac{n}{2}$ for which there is no positive anchor $\vec{p} \neq \vec{1}$ with $\vec{x} \leq \vec{p}$, leading to a contradiction: For each $\vec{p} \in P \setminus{\vec{1}}$, arbitrarily select some $i$ where $p_i=0$ and set $x_i = 1$, ensuring $\vec{x} \not \leq \vec{p}$. After this process, $\Delta(\vec{x}) \leq |P| \leq \frac{n}{2}$. Arbitrarily fixing more coordinates of $\vec{x}$ to $1$ so that $\Delta(\vec{x}) = \frac{n}{2}$ completes the construction.
\end{proof}

\section{Conclusion}
\label{sec:conclusion}

We have studied nearest neighbor representations of Boolean functions, proving new lower and upper bounds and devising connections to circuit complexity.
There are several questions we did not pursue. Can we find new connections between nearest neighbor complexity and other complexity measures of Boolean functions? For what classes of Boolean functions does polynomial nearest neighbor complexity translate to efficient learning algorithms? Similar ideas for establishing NP-hardness results of learning depth-two neural networks~\cite{blum1988training} can be shown to imply that properly learning Boolean functions representable by three real anchors is NP-hard. Devising efficient learnability results for nearest neighbor representations that circumvent worst-case barriers could be of interest. Finally, studying representations of Boolean functions using the well developed theory of \emph{approximate} nearest neighbor search~\citep{indyk1998approximate,kushilevitz1998efficient} could lead to new insights and more compact representations avoiding the curse of dimensionality.

\section{Acknowledgements} 

The second and third authors thank the Simons Institute for the Theory of Computing for their hospitality. Their collaboration on this project began during the Meta-Complexity workshop at the institute.

\vfill
\pagebreak

\bibliographystyle{plainnat}
\bibliography{Bib}

\vfill
\pagebreak

\appendix

\section{Omitted proofs}
\label{sec:proofs}

\subsection{Proof of Theorem \ref{th:mpPTFequalsNN}}

We break the proof of this theorem into two separate lemmas.

\begin{lemma}
\label{lem:NNtompPTF}
\[
\closNN \subseteq \mpPTF(\infty), \ \ \ 
\closHNN \subseteq \mpPTF(\poly(n))
\]
More precisely, any $\closNN$ representation with $m$ anchors is equivalent to an $\mpPTF$ with $m$ terms, and any $\closHNN$ representation with $m$ anchors in $\widetilde{n}$ dimensions is equivalent to an $\mpPTF$ with $m$ terms and maximum weight $\widetilde{n}$.

\end{lemma}

\begin{proof}
    The distance from $\vec{x} \in \{0,1\}^n$ to an anchor $\vec{p} \in \RR^n$ is a linear form in variables $\vec{x}$:
\begin{align*}
\sum_i (x_i-p_i)^2 &= \sum_i \left[x_i^2 - 2 p_i x_i + p_i^2\right] \\
&= \sum_i \left[(1-2p_i)x_i + p_i^2\right] \\
&= \langle \vec{1}-2\vec{p}, \vec{x} \rangle + ||\vec{p}||_2^2.
\end{align*}
We can observe that $\NN$ representations essentially compute $\ind[\min_{\vec{p} \in P}\Delta(\vec{x},\vec{p}) \leq \min_{\vec{q} \in N}\Delta(\vec{x},\vec{q})]$, which is an $\mpPTF$. subfunctions merely multiply coefficients and add constants to each linear form -- For example, $d(x_1x_10, p_1p_2p_3) = 2 \cdot (1-2p_1)x_1 + (p_3^2 + 2p_1^2)$.

In the case of $\HNN$, we have for all anchors that $\vec{p} \in \{0,1\}^n$ and $\Delta(\vec{x}, \vec{p})$ is a linear form with $\pm 1$ coefficients and positive constants bounded (in absolute value) by $n$. As a result, the weights in $\mpPTF$ are bounded by $n$ as well.
\end{proof}

\begin{lemma}
\label{lem:mpPTFtoNN}
\[
\mpPTF(\infty) \subseteq \closNN, \ \ \ 
\mpPTF(\poly(n)) \subseteq \closHNN
\]
More precisely, any $\mpPTF$ with $m$ terms has an $\closNN$ representation with $m$ anchors in $n+1$ dimensions.
Any $\mpPTF$ with $m$ terms and maximum weight $W$ has an $\closHNN$ representation with $m$ anchors in $\widetilde{n} = O(nW)$ dimensions. 
\end{lemma}

\begin{proof}
We start with the $\mpPTF(\poly(n))$ case.
Let $\ind[\min_{i \leq \ell_1} L_{1i}(\vec{x}) \leq \min_{j \leq \ell_2} L_{2j}(\vec{x})]$ be an arbitrary $\mpPTF(\poly(n))$. First make some pre-processing steps. First, multiply each linear form by $2$ and add one to the right-hand side, so that ties are won by the left-hand side. Second, we would like to make all coefficients positive. For this, while there exists a negative term $-a_{ijk}x_k$ (or constant $-\theta_{ij}$), just add $x_k$ (or $1$) to every linear form until all negative terms are eliminated. No coefficient (or constant) will increase by more than $W$. Third, we make all coefficients even by multiplying all linear forms by two. Finally, we add the same constant $\Theta$ (to be decided later) to all linear forms. Then, every linear form is equal to $L_{ij}(\vec{x}) = a_{ij1}x_1 + \cdots + a_{ijn}x_n + \theta_{ij} + \Theta$, for positive, even constants $a_{ijk}, \theta_{ij} \leq 8W$. 

Define $n$ \emph{block sizes} $t_1,\cdots,t_n$ by $t_k := \max_{i,j} a_{ijk}$ (i.e., the maximum coefficient of $x_k$ in any linear form). Also define $C = \Theta + \max_{i,j} \theta_{ij}$ and let $\widetilde{n}:=t_1+\cdots+t_n+C$. Inputs $\vec{x} \in \{0,1\}^n$ are mapped ($\rightsquigarrow$) to query points $\widetilde{\vec{x}} \in \{0,1\}^{\widetilde{n}}$ and linear forms $L_{ij}$ are mapped to anchors $\widetilde{\vec{p}}_{ij} \in \{0,1\}^{\widetilde{n}}$ such that $\Delta(\widetilde{\vec{x}},\widetilde{\vec{p}}_{ij}) = L_{ij}(\vec{x})$. In particular,
\[
\widetilde{\vec{x}} := \underbrace{x_1 \cdots \cdots x_1}_{t_1 \text{ many}} \cdots  \underbrace{x_n \cdots \cdots x_n}_{t_n \text{ many}} \cdot \underbrace{1 \cdots 1}_{C \text{ many}}
\]
and
\[
\widetilde{\vec{p}}_{ij} = \underbrace{0 \cdots \cdots 0}_{(t_1+a_{ij1})/2}\underbrace{1 \cdots \cdots 1}_{(t_1-a_{ij1})/2} \cdots  \underbrace{0 \cdots \cdots 0}_{(t_n+a_{ijn})/2}\underbrace{1 \cdots \cdots 1}_{(t_n-a_{ijn})/2} \cdot  \underbrace{0 \cdots \cdots 0}_{z_{ij}}\underbrace{1 \cdots \cdots 1}_{C - z_{ij}}
\]
where $z_i$ will be chosen momentarily. (Let $P = \{\widetilde{\vec{p}}_{1j}\}_{j \leq \ell_1}$ and  $N = \{\widetilde{\vec{p}}_{2j}\}_{j \leq \ell_2}$.) The distance between $\widetilde{\vec{x}}$ and $\widetilde{\vec{p}}_{ij}$ is equal to
\begin{align*}
\Delta(\widetilde{\vec{x}}, \widetilde{\vec{p}}_{ij}) &= z_{ij} + \sum_{k} \left( \frac{t_k+a_{ijk}}{2} \right) x_k + \left( \frac{t_k-a_{ijk}}{2}\right)(1-x_k)  \\
&=  z_{ij} + \langle \vec{a_{ij}}, \vec{x} \rangle + \sum_{k} \left( \frac{t_k-a_{ijk}}{2}\right)
\end{align*}

Now let $z_{ij} = \Theta + \theta_{ij} - \sum_{k} \left( \frac{t_k-a_{ijk}}{2}\right)$ so that $\Delta(\widetilde{\vec{x}}, \widetilde{\vec{p}}_{ij}) =  \langle \vec{a_{ij}}, \vec{x} \rangle + \Theta + \theta_{ij}$. This is valid (i.e., $z_{ij}$ is a non-negative integer) if we choose a large enough value for $\Theta$: The minimal value of $\Theta$ such that $z_{ij} \geq 0$ for all $i,j$ is
\[
\Theta = \max_{i,j}  \left( \sum_{k} \left( \frac{t_k-a_{ijk}}{2}\right) - \theta_{ij}\right) \leq \sum_k \frac{t_k}{2} \leq 4nW.
\]

Thus, for $\Theta = 4nW$, we may always choose $0 \leq z_{ij} \leq \Theta + \theta_{ij} \leq C$. Observe that $\vec{x} \rightsquigarrow \widetilde{\vec{x}}$ by duplicating each $x_i$ at most $8W$ times and introducing at most $4nW + 8W$ constant variables. Thus, the original $\mpPTF$ is equivalent to a subfunction of an $\HNN$ representation with $m$ anchors at most $4nW + 8W$ dimensions.

For the $\mpPTF(\infty)$ case, the same method applies, only now we do not need to increase the dimension that much. All coefficients can be realized by choosing anchors $\vec{p}_{ij} = (1-\vec{a}_{ij})/2$ and all constants $\theta_{ij}$ can be corrected using one additional dimension. 
\end{proof}

From this we can also deduce the following:

\begin{theorem}
Any function with an $m$-anchor $\NN$ representation with bit-complexity $O(\log n)$ is equivalent to an $\mpPTF(\poly(n))$ with $m$ terms. Any function of $n$ inputs with $\mpPTF(\poly(n))$ representation with $m$ terms is equivalent to a subfunction of a function of $n+1$ inputs with an $m$-anchor $\NN$ representation with bit-complexity $O(\log n)$. 
\end{theorem}

\begin{proof}
Observe that in Lemmas~\ref{lem:NNtompPTF} and~\ref{lem:mpPTFtoNN} -- for $\NN$ and $\mpPTF(\infty)$ -- the bit-complexity of $\NN$ and the logarithms of weights of $\mpPTF$ are linearly related.
\end{proof}

\subsection{Proof of Corollary \ref{cor:HNNneqHNNbar}}

\begin{proof}
It is shown in \cite{hajnal2022nearest} that $\XOR$ has a unique $\HNN$ representation with $2^n$ anchors. Furthermore, it is established in \cite{hansen2015polynomial} that $\XOR \in \mpPTF(\poly(n))$: In particular, $\XOR(\vec{x}) = 1$ if and only if $\min \left\{L_0(\vec{x}),L_2(\vec{x}),\cdots\right\} \leq \min \left\{L_1(\vec{x}),L_3(\vec{x}),\cdots\right\}$ where $L_i(\vec{x}) = i^2 - 2i \cdot (x_1+\cdots+x_n)$.
\end{proof}

\subsection{Proof of Corollary \ref{cor:signrankLB}}

\begin{proof}
It was shown by \cite{hajnal2022nearest} that the $\NN$ complexity of a Boolean function $f$ is bounded below by the \emph{sign-rank} of $f$, and this can be easily extended to $\closNN$ through Theorem \ref{th:mpPTFequalsNN}: The number of terms in an $\mpPTF$ computing $f$ is also bounded below by the sign-rank of $f$, by \cite{hansen2015polynomial}.

\cite{forster2002linear} and \cite{razborov2010sign} respectively establish that the sign rank of $\IP$ is equal to $2^{n/2}$ and the sign rank of $f_n$ is $2^{\Omega(n)}$.
\end{proof}

\subsection{Proof of Lemma \ref{lem:PTF-to-LDL}}

\begin{proof}
    Consider a function $f \in \mpPTF(\poly(n))$ and let $ \ind[\min_{i \leq \ell_1} L_{1i}(\vec{x}) \leq \min_{j \leq \ell_2} L_{2j}(\vec{x})]$ be its representation. We can assume that all possible values of all linear forms are distinct. For this it is enough to multiply all forms by $\ell_1+\ell_2$ and to add to each form it's own unique remainder modulo $\ell_1+\ell_2$. 

    Observe that all linear forms obtain only polynomially many variables (since there output is polynomially bounded in absolute value). Denote possible values of the form $L_{ij}$ by $a_{ij1}, \cdots, a_{ijt}$ for some $t$ polynomially bounded in $n$. Note that, for different linear forms, the number of the values obtained might be not the same. To simplify the notation we assume that we add several equal values to the list to make them all of equal size $t$. 

    Now we are ready to produce the decision list. Let $c_1=1$ and $c_2=0$. We consider each $a_{ijk}$ in increasing order and query if $L_{ij}(\vec{x}) \leq a_{ijk}$. If so, we output $c_i$. If not, we proceed to the next $a_{ijk}$.

    This decision list computes $f$ since we are just looking for the minimal value of a linear form among all possible values of the forms.
\end{proof}

\subsection{Consequences of Theorem \ref{th:kNNmpPTF}}

\begin{corollary}
\label{cor:kNNinNNbar}
Any Boolean function with a $\closkNN$ representation with $m$ anchors has an $\closNN$ representation with $\binom{m}{k}$ anchors. (Similarly, Boolean function with a $\closkHNN$ representation with $m$ anchors has an $\closHNN$ representation with $\binom{m}{k}$ anchors.)
\end{corollary}

\begin{remark} \label{rem:non-boolean}
Theorem~\ref{th:kNNmpPTF} and Corollary~\ref{cor:kNNinNNbar} can be extended to non-Boolean inputs. More precisely, the same statements are true over any finite domain $D \subseteq \mathbb{R}^n$. For this we can express (squared) distances to anchors as quadratic forms, for each subset of distances of size $k$ consider the average of these distances and represent them as a distance to a new anchor. We still need to add an extra dimension to absorb constant terms.
\end{remark}

\begin{remark}
Theorem~\ref{th:kNNmpPTF} and Corollaries~\ref{cor:kNNinNNbar} and~\ref{cor:kNNsignrank} can be extended to the case of weighted $\kNN$. Indeed, in Theorem~\ref{th:kNNmpPTF}, instead of sums of linear forms we will have weighted sums. This will require $\binom{m}{k} \cdot k! = \frac{m!}{(m-k)!}$ terms in the $\mpPTF$ representation. If the weights in the weighted $\kNN$ representation are small and the bit-complexity of anchors is small, this results in a $\closHNN$ representation and if there are no restrictions of weights and bit-complexity, we get $\closNN$ representation. The proof of Corollary~\ref{cor:kNNsignrank} still works despite the increase of the number of anchors to $\frac{m!}{(m-k)!}$.
\end{remark}

\subsection{Proof of Theorem \ref{th:kNNequalskSTAT}}

We first make the following general observation: \cite{patrick1970generalized} show that finding the $k$'th nearest positive anchor and $k$'th nearest negative anchor and classifying based on which is closest is equivalent to computing a $(2k-1)$-nearest neighbors representation. This fact can be generalized, considering the closure of $\kNN$. 

\begin{lemma} \label{lem:kSTAT vs. Majority}
Let $A$ and $B$ be two sets of numbers and let $S$ be the $k$ smallest elements of $A \cup B$. Then, 
\[
|A \cap S| \geq |B \cap S| \iff A_{(t)} < B_{(t)}
\]
where $t=\left\lfloor\frac{k+1}{2} \right\rfloor$. (As in $\kNN$, we assume $S$ exists and is unique).
\end{lemma}

\begin{proof}
    $A$ contains a majority of the elements in $S$ if and only if $|A \cap S| \geq t$. This happens if and only if the $t$'th smallest element in $A$ is smaller than the $t$'th smallest element in $B$.
\end{proof}

We now proceed with the proof of Theorem \ref{th:kNNequalskSTAT}.

\medskip
\begin{proof}
    For the inclusion $\closkNN \subseteq \kSTAT$, consider any function $f$ in $\closkNN$. It is a subfunction of some function $g$ with a $\kNN$ representation $P \cup N$. As in Lemma \ref{lem:NNtompPTF}, the distances between $\vec{x}$ and each anchor are linear forms $A = \{L_{1}(\vec{x}),\cdots,L_{|P|}(\vec{x})\}$ and $B = \{R_1(\vec{x}),\cdots,R_{|N|}(\vec{x})\}$ which we assume have integer coefficients by the usual finite precision argument. By definition $g(\vec{x}) = 1$ if and only if the set $S$ of $k$-nearest neighbors satisfies $|P \cap S| \geq |N \cap S|$. By Lemma~\ref{lem:kSTAT vs. Majority}, this happens if and only if $A_{(t)} < B_{(t)}$, taking $t=\left\lfloor\frac{k+1}{2} \right\rfloor$. Hence, $g \in \kSTAT$. As $\kSTAT$ is closed under taking subfunctions, $f \in \kSTAT$ as well.

    For the inclusion $\kSTAT \subseteq \closkNN$, assume that $f$ has a $\kSTAT$ representation. By adding dummy linear forms we can have $k_l=k_r$. By Lemma~\ref{lem:kSTAT vs. Majority}, the inequality~\eqref{eq:kSTAT} holds if and only if the $2 k_l - 1$ smallest linear forms consist of more linear forms from the left-hand side than the right. Representing each inequality by an anchor, we obtain a representation of the same function in $\closkNN$.

    The case of $\closkHNN$ and $\kSTATlw$ is analogous.
\end{proof}

\subsection{Proof of Theorem \ref{thm:altern-k-stat}}

\begin{proof}
    Suppose a Boolean function $f$ has a representation $\{L_1, \cdots, L_p\}$ satisfying \eqref{eq:altkSTAT} for some function $\lab$ and integer $k$. We will show that $f \in \kSTAT$. First, we assume that all coefficients in all linear form are integers and ensure that all values of all linear forms are distinct and even. For this, multiply all forms by $2p$ and shift each form by its own even remainder modulo $2p$.
    
    For each $i \leq p$, we add one linear form to each side of \eqref{eq:kSTAT}. If $\lab(i) = 1$, then place the form $L_i(\vec{x})$ on the left-hand side and $L_i(\vec{x})+ 1$ on the right. If $\lab(i)=0$, put the $L_i(\vec{x})$ on right-hand side and $L_i(\vec{x})+1$ on the left. It is easy to see that the $k$'th statistics in the left and right-hand sides of the resulting $\kSTAT$ representation are $L_i(\vec{x})$ and $L_i(\vec{x}) + 1$ (not necessarily in that order), where $L_i(\vec{x})$ is the $k$'th statistic of the original representation. Hence, the inequality in \eqref{eq:kSTAT} holds if and only if $\lab(i)=1$.

    For the other direction, assume we have a function $f \in \kSTAT$ given by \eqref{eq:kSTAT}. We again assume that all coefficients are integers and all values of all linear forms are distinct. Now we construct the required representation of $f$. For each form $L_i$ we add to the representation the forms 
    \[
    L_{ij}(\vec{x}) :=L_i(\vec{x}) + \frac{j}{k_l+k_r}
    \]
    for all $j \in \{0, 1,\cdots, k_l+k_r-1\}$, and for each form $R_i$ we add to the representation the forms 
    \[
    R_{ij}(\vec{x}) := R_i(\vec{x}) + \frac{j}{k_l+k_r+1}
    \]
    for all $j = \{0,1, \cdots, k_l+k_r\}$. (That is, we have $k_l+k_r$ copies of each form $L_i$ and $k_l+k_r+1$ copies of each form $R_i$). To each $L_{ij}$, $R_{ij}$ we assign the label $0$ if $j < k_l$, and $1$ if $j \geq k_l$. Finally, we set $k=(k_l+k_r-1)(k_l+k_r+1) + 1$.

    Now, observe that the inequality~\eqref{eq:kSTAT} holds if and only if, among the $k_l+k_r-1$ smallest forms, there are at least $k_l$ forms $L_i$. Assume that there are precisely $a$ forms $L_{i}$ and $b$ forms $R_{i}$. In particular, $a+b = k_l + k_r -1$. Then, in the new representation, these linear forms give us
    \[
    a(k_l + k_r) + b(k_l + k_r +1) = (a+b) (k_l + k_r +1) - a = (k_l + k_r -1) (k_l + k_r +1) - a
    \]
    smallest forms. By construction, the next smallest forms are either $L_{i0} \leq \cdots \leq L_{i(k_l+k_r)}$ or $R_{i0} \leq \cdots \leq R_{i(k_l+k_r+1)}$ for some $i$. Thus, the $k$'th smallest form is either $L_{ia}$ or $R_{ia}$ and its label is $1$ if and only if $a \geq k_l$ as desired.
\end{proof}

\subsection{Proof of Theorem \ref{th:SYMMAJinkSTAT}}

\begin{proof}
    Suppose we are given a function $f \in \SYM \circ \MAJ$ and a circuit computing it. We are going to construct a $\kSTATlw$ representation of $f$ in the form given by Theorem~\ref{thm:altern-k-stat}.
    
    We can assume that all $\MAJ$ gates in the circuit have the same threshold $t=0$. For this we can just add dummy variables and fix them to constants. Denote the linear forms for $\MAJ$ gates by $L_1, \cdots, L_s$ (all weights are integers) and denote by $g\colon \{0,1\}^s \to \{0,1\}$ the symmetric function at the top of the circuit. Here, $s$ is the size of the circuit. Now, construct a $\kSTATlw$ representation with the following linear forms:
    \begin{equation} \label{eq:forms-for-sym-maj}
    (s+2)L_1(\vec{x}), \cdots, (s+2)L_s(\vec{x}), 1, 2, \cdots, s+1.
    \end{equation}
    That is, we multiply each linear form by $(s+2)$ and add $(s+1)$ constant linear forms with values $1, \cdots, s+1$. We let $k = s+1$.

    It is easy to see that the $k$'th statistic of~\eqref{eq:forms-for-sym-maj} is always one of the constant linear forms. It is the form $i$ if and only if $i-1$ of the linear forms among $L_1, \cdots, L_s$ are positive. We assign label $1$ to the form $i$ if and only if $g(\vec{x})=1$ for inputs of weight $i-1$. As a results we get the desired representation for $f$ and show that $f \in \kSTATlw$.
\end{proof}

\begin{remark}
    The well-known argument that shows $\MAJ \circ \THR = \MAJ \circ \MAJ$ (see~\cite{goldmann1992majority}) can be straightforwardly adapted to show that $\SYM \circ \THR = \SYM \circ \MAJ$. Thus, $\SYM \circ \THR \subseteq \kSTATlw$ follows from Theorem \ref{th:SYMMAJinkSTAT} as well.    
\end{remark}

\subsection{Proof of Theorem \ref{th:SYMANDinkNN}}

\begin{proof}
    First, as a warm-up, we show that $\IP \in \kNN$. Recall that $\IP(\vec{x},\vec{y}) = \bigoplus_{i=1}^n (x_i \wedge y_i)$. Denote by $\vec{a} = (\frac 12, \cdots, \frac 12)$ an $2n$-dimensional vector with $\frac 12$ in each coordinate. Note that $\Delta(\vec{a}, (\vec{x},\vec{y})) = \frac{n}{2}$ for all $(\vec{x},\vec{y}) \in \{0,1\}^{2n}$.

    For each $i = 1, \cdots, n$ and $j=0,1$ introduce an anchor 
    $\vec{p_{ij}} = \vec{a} + (\frac 12 - \frac j4) (\vec{e_i} + \vec{e_{i+n}})$. If for some $(\vec{x},\vec{y})$ we have $x_i=y_i=1$, then
    \[
        \Delta((\vec{x},\vec{y}), \vec{p_{ij}}) \leq \frac{n}{2} - 2 \left(\frac{1}{4} - \frac{1}{16} \right) = \frac{n}{2} - \frac{3}{8}.
    \]
    If, on the other hand, $x_i=0$ or $y_i=0$, then
    \[
        \Delta((\vec{x},\vec{y}), \vec{p_{ij}}) \geq \frac{n}{2} - \left(\frac{1}{4} - \frac{1}{16}\right) - \left(\frac 14 - \frac{9}{16}\right) = \frac{n}{2} + \frac{1}{8} > \frac{n}{2}.
    \]
    For each $i= 1, \cdots, n+1$ and $j = 0, 1$ and $l=0,1$ introduce an anchor $\vec{q_{i,j,l}} = \vec{a} + (-1)^{l} \frac{2i+j}{8n} \vec{e_1}$. For $(\vec{x},\vec{y})$ with $x_1=1$ it is not hard to see that 
    \begin{align*}
    \frac{n}{2} - \frac{3}{8} <& \Delta((\vec{x},\vec{y}), \vec{q_{n+1,1,0}}) < 
    \Delta((\vec{x},\vec{y}), \vec{q_{n+1,0,0}}) < \cdots <\\ <& 
    \Delta((\vec{x},\vec{y}), \vec{q_{1,1, 0}}) <
    \Delta((\vec{x},\vec{y}), \vec{q_{1,0, 0}}) <
    \frac{n}{2}
    \end{align*}
    and $\Delta((\vec{x},\vec{y}), \vec{q_{i,j,1}}) > \frac n2$ for all $i,j$.
    The situation is symmetric for $x_1=0$.
    We assign label $j$ to the anchor $\vec{p_{ij}}$. 
    We assign label $1$ to the anchor $\vec{q_{ijl}}$ iff $i + j$ is odd. 
    We let $k = 2n+1$.

    It is easy to see that for a given $(\vec{x},\vec{y})$ among the $k$ closest anchors we have all pairs of anchors $\vec{p_{i0}}, \vec{p_{i1}}$ for all $i$ such that $x_i=y_i=1$. Denote the number of such $i$ by $t$. Also among the $k$ closest anchors we will have pairs of anchors $\vec{q_{i,0,l}}, \vec{q_{i,1,l}}$ for an appropriate $l$ and for $i = n+1, \ldots, t+2$. In each of these pairs the labels of anchors are opposite and they cancel out when we compute the majority. Finally, one last anchor we will have among the $k$ closest anchors is $\vec{q_{t+1,1,l}}$. The label of this anchor determines the majority among the $k$ closest anchors and it is 1 iff $t$ is odd.
    As a result, we get the desired representation for $\IP$ with $6n+4$ anchors.

    Now we extend this argument to any function in $\SYM \circ \AND$. Consider a function $f(\vec{x}) = g(f_1(\vec{x}),\cdots, f_s(\vec{x}))$, where each $f_i$ has the form
    \[
    f_{i}(\vec{x}) = \left( \bigwedge_{i \in S_i} x_i \right) \wedge \left( \bigwedge_{i \in T_i} \neg x_i \right)
    \]
    for some disjoint $S_i, T_i \subseteq [n]$.

    For each $f_i$ we let $\epsilon_{i1} > \epsilon_{i0} >0$ be a couple of parameters to be fixed later. We introduce a pair of anchors $\vec{p_{i1}}, \vec{p_{i0}}$ in the following way: we let the $j$th coordinate of $\vec{p_{ij}}$ to $\frac 12$ if $i \notin S_i \cup T_i$, to $\frac 32 - \epsilon_{ij}$ for $i \in S_i$ and to $ - \frac 12 + \epsilon_{ij}$ for $i \in T_i$. It is easy to see that for $\vec{x}$ such that $f_i(\vec{x})=1$ we have $\Delta(\vec{p_{ij}},\vec{x}) = \frac{n}{4} - |S \cup T|(\epsilon_{ij} - \epsilon_{ij}^2)$ and for $\vec{x}$ such that $f_i(\vec{x})=0$ we have $\Delta(\vec{p_{ij}},\vec{x}) \geq \frac{n}{4} - \left(|S \cup T|-1\right)(\epsilon_{ij} - \epsilon_{ij}^2) + 1$ for $\epsilon_{ij} < \frac 12$. We fix $\epsilon_{ij}$ in such a way that $\frac{n}{4} - |S \cup T|(\epsilon_{ij} - \epsilon_{ij}^2) < \frac n4 - \frac 12$ and $\frac{n}{4} - \left(|S \cup T|-1\right)(\epsilon_{ij} - \epsilon_{ij}^2) + 1 > \frac n4 + \frac 12$. We set $\lab(\vec{p_{ij}}) = j$.

    We construct anchors $\vec{q_{ijl}}$ for $i=1, \cdots, s+1$ and $j=0,1$ the same way as above and assign $\lab(\vec{q_{i1l}})$ to be equal to $g(\vec{y})$ for $\vec{y}$ of weight $i-1$ and $\lab(\vec{q_{i0l}})$ to be the opposite. We let $k=2s+1$. The same argument as for $\IP$ shows that we get the desired representation of $f$ with $6s+4$ anchors.
\end{proof}

\subsection{Proof of Theorem \ref{thm:eldl-in-kstat}}
\begin{proof}
    Consider a function $f \in \ELDL$ and suppose the linear forms in its representation are $L_1, \cdots, L_s$. Here $L_i$ corresponds to the $i$'th query. As in the proof of Theorem \ref{th:SYMMAJinkSTAT}, we can assume that all thresholds in all linear forms are 0.

    We are going to construct a representation for $f$ of the form provided by Theorem~\ref{thm:altern-k-stat}.
    We add to this representation the following linear forms:
    \[
        (s+1)L_1, -(s+1)L_1, (s+1)L_2 + 1, -(s+1)L_2 - 1, \cdots, (s+1)L_s + s-1, - (s+1)L_s - s + 1.
    \]
    That is, for each form $L_i$ in $\ELDL$ representation, we add the two forms $(s+1)L_i + (i-1)$ and $-(s+1)L_i - (i-1)$.
    We set $k=s$.

    Assume that for some $x$ we have $L_i(\vec{x})=0$ and all previous linear forms are non-zero. We than have that $(s+1)L_i(\vec{x}) + i-1 = i-1$. 
    It is not hard to see that for $j<i$ we have that among forms $(s+1)L_j(\vec{x}) + j-1$ and $-(s+1)L_j(\vec{x}) - j+ 1$  exactly one is greater than $i-1$: it is the first one if $L_j(\vec{x})>0$ and the second one if $L_j(\vec{x})<0$. For $j>i$ in a similar way we can see that among the forms $(s+1)L_j(\vec{x}) + j-1$ and $-(s+1)L_j(\vec{x}) - j+ 1$ exactly one is greater than $i-1$: it is the first one if $L_j(\vec{x})\geq0$ and the second one if $L_j(\vec{x})<0$. As a result there are exactly $s-1$ forms that are greater than $(s+1)L_i(\vec{x}) + i-1$. We assign to this form the same label $L_i(\vec{x})$ has in $\ELDL$. From this it follows that the constructed representation computes the same function.

    Clearly, the coefficients in the constructed form are polynomially related to the coefficients in the original forms. Thus, the same proof gives $\ELDLlw \subseteq \kSTATlw$.
\end{proof}

\begin{remark}
    Note that decision lists are computable in $\AC^0$ and thus can be computed by quasi-polynomial-size $\SYM \circ \AND$ circuits. As a result, $\ELDL$ can be computed by quasi-polynomial-size circuit in $\SYM \circ \AND \circ \ETHR = \SYM \circ \ETHR = \SYM \circ \MAJ$, where the second equality follows since $\ETHR$ is closed under $\AND$ operation. Still, Theorem~\ref{thm:eldl-in-kstat} gives a polynomial reduction that translates to the case of small coefficients. 
\end{remark}

\subsection{Proof of Theorem \ref{th:CNFnotinHNN}}

Such constructions are likely known; we outline a simple one for completeness.

\begin{lemma}
\label{lem:CNFconstruction}
For any even integer $k>0$, there exists a CNF with $n$ variables and $n2^{k-o_k(k)}/k$ clauses with $2^{(1-o_k(1))n}$ components.  
\end{lemma}

\begin{proof}
Assume $k$ divides $n$. Divide the set of variables to $n/k$ disjoint sets $S_1,\cdots,S_{n/k}$ of size $k$. For each set $S_i$, define a CNF $C_i$ which evaluates to $1$ if and only if exactly half of the variables in $S$ are equal to $1$. This can be achieved with $\binom{k}{k/2}=2^{k-o_k(k)}$ clauses. 

Then, the CNF $C=C_1 \wedge \cdots \wedge C_{n/k}$ has exactly $\binom{k}{k/2}^{n/k}=2^{(1-o_k(1))n}$ satisfying assignments, and the Hamming distance between any two of such assignments is at least $2$. Thus, each of them constitutes a component.
\end{proof}

Hence, Theorem \ref{th:CNFnotinHNN} follows from Lemmas \ref{lem:HNNcomponents} and \ref{lem:CNFconstruction} by taking $k$ to be a constant independent of $n$. It is easy to extend the construction above to odd $k$. We omit the simple details.

\section{Circuits computing nearest neighbors}
\label{sec:circuitUB}

In this section we describe a straightforward construction of a \emph{depth-three} circuit computing $\HNN$ and then compress it to \emph{depth-two} at the cost of exponential weights. The folklore result of \cite{murphy1992nearest} is that any $\NN$ representation with $m$ anchors can be computed by a depth three threshold circuit with size $O(m^2)$. A short proof can be found in \cite{kilic2023information}.

\begin{theorem}[\cite{murphy1992nearest}]
\label{th:canonicalHNN}
\
\begin{itemize}
   \item $\NN \subseteq \OR \circ \AND \circ \THR, \ \ \AND \circ \OR \circ \THR$
   \item $\HNN \subseteq \OR \circ \AND \circ \MAJ, \ \ \AND \circ \OR \circ \MAJ$
\end{itemize}
Namely, every $\NN$ ($\HNN$) representation is computed by a depth-three $\AC^0 \circ \THR \ (\MAJ)$ circuit with size $|P||N| + \min\{|P|,|N|\} + 1$.
\end{theorem}

Note that the only difference between the circuits for $\HNN$ and $\NN$ is that the first-level threshold gates are guaranteed to have polynomial weights (in the case of $\HNN$). It turns out that the size of the $\HNN$ circuit can be improved (when $n \ll |P| + |N|$).

\begin{lemma}
\label{lem:HNNslice}
\[
\HNN \subseteq \OR \circ \AND \circ \MAJ
\]
In particular, every $\HNN$ representation with $m$ anchors is computed by an $\OR \circ \AND \circ \THR$ circuit with size $(n+1)m + (n+1)|P| + 1$.
\end{lemma}

\begin{proof}
Note that $\ind[\Delta(\vec{x}, \vec{p}) \leq i]$ is computed by a threshold gate $f^{\vec{p}}_{\leq i}(\vec{x})$ defined by $\vec{w} = \vec{p} - \overline{\vec{p}}$ and $\theta = \Delta(\vec{p}) - i$. (And similarly $\ind[\Delta(\vec{x}, \vec{p}) \geq i]$.) Suppose $f$ has an $\HNN$ representation $P\cup N$. Then,
\[
f(\vec{x}) = \bigvee_{\begin{smallmatrix} i \leq n \\ \vec{p} \in P \end{smallmatrix}}\left( f^{\vec{p}}_{\leq i}(\vec{x}) \wedge \bigwedge_{\vec{q} \in N} f^{\vec{q}}_{\geq i}(\vec{x})\right) 
\]
\end{proof}

Note that the threshold circuits from Theorem \ref{th:canonicalHNN} and Lemma \ref{lem:HNNslice} have size $O(m^2)$ and $O(mn)$ respectively. In fact, the latter circuit can be compressed to a depth-two threshold circuit with exponential weights. 

\begin{theorem}
\label{th:HNNdepth2}
\[
\HNN \subseteq \THR \circ \MAJ.
\]
Namely, every $\HNN$ representation with $m$ anchors is computed by a threshold of $2 n m$ majority gates.
\end{theorem}

\begin{proof}
The first level will consist of $2mn$ gates $f^{\vec{p}}_{\leq i}$, $f^{\vec{p}}_{\geq i}$ which output $1$ if and only if $\Delta(\vec{x},\vec{p}) \leq i$ and $\Delta(\vec{x},\vec{p}) \geq i$, respectively, for $1 \leq i \leq n$. Define the sum
\[
g_{i}^p(\vec{x}) := f^{\vec{p}}_{\leq i}(\vec{x}) + f^{\vec{p}}_{\geq i}(\vec{x}) - 1
\]
and note that $g_{i}^p(\vec{x}) =   \ind[\Delta(x,p) = i]$. We can then write the output gate as
\[
h(\vec{x}) = \ind \left( \sum_{p \in P, i \leq n} m^{3(n-i)+1} g_{i}^p(\vec{x}) - \sum_{q \in N, i \leq n} m^{3(n-i)} g_{i}^q(\vec{x}) \geq 0 \right).
\]

If some positive anchor is at distance at most $j$ and all negative anchors are at distance at least $j$ to $\vec{x}$, then 
\[
\sum_{p \in P, i \leq n} m^{3(n-i)+1} g_{i}^p(\vec{x}) \geq m^{3(n-j)+1} \geq \sum_{q \in N, i \leq n} m^{3(n-i)} g_{i}^q(\vec{x}).
\]
Conversely, if some negative anchor is at distance at most $j$ and all positive anchors are at distance at least $j+1$, then 
\[
\sum_{p \in P, i \leq n} m^{3(n-i)+1} g_{i}^p(\vec{x}) \leq m^{3(n-j)-1} < m^{3(n-j)} \leq \sum_{q \in N, i \leq n} m^{3(n-i)} g_{i}^q(\vec{x}).
\]
\end{proof}

\begin{remark}
Theorem \ref{th:HNNdepth2} can be obtained through Theorem \ref{th:mpPTFequalsNN}, as a consequence of the following result derived from \cite{hansen2015polynomial}. We include the direct construction to avoid the slight increase in circuit size.
\end{remark}

\begin{lemma}
\label{lem:mpPTFdepthtwo}
\[
\mpPTF(\poly(n)) \subseteq \THR \circ \MAJ
\]
Every $\mpPTF$ with $\ell$ terms and maximum weight $W$ is computed by a linear threshold of at most $4 \cdot W \ell \log \ell$ majority gates.
\end{lemma}
\begin{proof}
Let $\PTF_{1,2}$ refer to Boolean functions (over $\{1,2\}$) equal to the sign of an $n$-variate polynomial. \cite{hansen2015polynomial} prove that any $\PTF_{1,2}$ with $\ell$ terms and degree at most $d$ is computed by a linear threshold (with exponential weights) of at most $2 \ell d$ majority gates (replacing $\{1,2\}$ with $\{0,1\}$), and any $\mpPTF$ with $\ell$ terms and maximum weight $W$ can be represented by a $\PTF_{1,2}$ with $\ell$ terms and degree at most $2 W \log \ell$.

\end{proof}

\begin{remark}
It is not hard to see that the circuits constructed in this section are \emph{polynomial-time uniform}; they can be generated by a Turing machine given the set of anchors in polynomial time.
\end{remark}

\end{document}